\documentclass{birkjour}

\usepackage{mathtools}
\usepackage{amssymb}
\usepackage{amsthm}
\usepackage{xcolor}
\usepackage{tikz}
\usepackage{fullpage}
\usepackage{url}

\DeclareSymbolFont{bbold}{U}{bbold}{m}{n}
\DeclareSymbolFontAlphabet{\mathbbold}{bbold}

\newcommand{\R}{\mathbb{R}}
\newcommand{\C}{\mathbb{C}}



\renewcommand{\Re}{\operatorname{Re}}

\renewcommand{\l}{\mathrm{l}}
\renewcommand{\r}{\mathrm{r}}
\DeclareMathOperator{\tr}{tr}
\DeclareMathOperator{\Tr}{Tr}

\providecommand{\scpr}[2]{\left( #1 \,\middle|\, #2 \right)}
\providecommand{\dupa}[2]{\left\langle #1 , #2 \right\rangle}
\renewcommand{\sp}{\scpr}
\newcommand{\from}{\colon}

\let\phi\varphi

\let\leq\leqslant

\let\geq\geqslant

\makeatletter
\def\@row#1,{#1\@ifnextchar;{\@gobble}{&\@row}}
\def\@matrix{%
    \expandafter\@row\my@arg,;%
    \@ifnextchar({\\ \get@in@paren{\@matrix}}{\after@matrix}%
    }
\def\matrixtype#1#2#3{%
    \ifmmode\def\after@matrix{\end{#2}\right#3}%
    \else\def\after@matrix{\end{#2}\right#3$}$\fi
    \left#1\begin{#2}\get@in@paren{\@matrix}%
    }
\def\@column#1,{#1\@ifnextchar;{\@gobble}{\\ \@column}}
\newcommand\vect{}
\def\svect(#1){\left(\begin{smallmatrix}\@column#1,;\end{smallmatrix}\right)}
\def\vect{\get@in@paren{\@vect}}
\def\@vect{\left(\begin{matrix}\expandafter\@column\my@arg,;\end{matrix}\right)}
\def\get@in@paren#1({\def\my@arg{}\def\my@rest{}\def\after@get{#1}\get@arg}
\let\e@a\expandafter
\def\get@arg#1){\e@a\kl@test\my@rest#1(;}
\def\kl@test#1(#2;{\e@a\def\e@a\my@arg\e@a{\my@arg#1}%
                   \ifx:#2:\let\my@exec\after@get
                   \else\let\my@exec\get@arg
                        \e@a\def\e@a\my@arg\e@a{\my@arg(}%
                        \def@rest#2;%
                   \fi\my@exec}
\def\def@rest#1(;{\def\my@rest{#1\kl@zu}}
\def\kl@zu{)}

\makeatletter
\newcommand\MyPairedDelimiter{%
  \@ifstar{\My@Paired@Delimiter{{}}}
          {\My@Paired@Delimiter{}}%
}
\newcommand\My@Paired@Delimiter[4]{%
  \newcommand#2{%
    \@ifstar{\start@PD{#1}{\delimitershortfall=-1sp}{#3}{#4}}
            {\start@PD{#1}{}{#3}{#4}}%
  }%
}
\newcommand\start@PD[5]{%
  #1\mathopen{\mathpalette\put@delim@helper{\put@delim{#2}{#3}{.}{#5}}}%
  #5%
  \mathclose{\mathpalette\put@delim@helper{\put@delim{#2}{.}{#4}{#5}}}%
}
\newcommand\put@delim@helper[2]{%
  \hbox{$\m@th\nulldelimiterspace=0pt #2#1$}%
}
\newcommand\put@delim[5]{%
  \setbox\z@\hbox{$\m@th#5{#4}$}%
  \setbox\tw@\null
  \ht\tw@\ht\z@ \dp\tw@\dp\z@
  #1#5%
  \left#2\box\tw@\right#3%
}

\makeatother
\MyPairedDelimiter*{\abs}{\lvert}{\rvert}
\MyPairedDelimiter*{\norm}{\lVert}{\rVert}
\MyPairedDelimiter{\set}{\{}{\}}

\newcommand\llim{
\mathchoice{\vcenter{\hbox{${\scriptstyle{-}}$}}}
{\vcenter{\hbox{$\scriptstyle{-}$}}}
{\vcenter{\hbox{$\scriptscriptstyle{-}$}}}
{\vcenter{\hbox{$\scriptscriptstyle{-}$}}}}
\newcommand\rlim{
\mathchoice{\vcenter{\hbox{${\scriptstyle{+}}$}}}
{\vcenter{\hbox{$\scriptstyle{+}$}}}
{\vcenter{\hbox{$\scriptscriptstyle{+}$}}}
{\vcenter{\hbox{$\scriptscriptstyle{+}$}}}}

\theoremstyle{plain} 
\newtheorem{theorem}{Theorem}[section]
\newtheorem{corollary}[theorem]{Corollary}
\newtheorem{lemma}[theorem]{Lemma}

\theoremstyle{definition}

\newtheorem*{definition}{Definition}
\newtheorem{remark}[theorem]{Remark}

\usepackage{enumitem}
\renewcommand{\labelenumi}{(\alph{enumi})}
\renewcommand{\labelenumii}{(\roman{enumii})}

\setenumerate[1]{nolistsep} 
\setenumerate[2]{nolistsep} 

\newcommand{\Hmm}[1]{\leavevmode{\marginpar{\tiny%
$\hbox to 0mm{\hspace*{-0.5mm}$\leftarrow$\hss}%
\vcenter{\vrule depth 0.1mm height 0.1mm width \the\marginparwidth}%
\hbox to 0mm{\hss$\rightarrow$\hspace*{-0.5mm}}$\\\relax\raggedright #1}}}

\begin{document}

\title{The linearized Korteweg-de-Vries equation on general metric graphs}

\author{Christian Seifert}
\address{%
Technische Universit\"at Hamburg\\
Insitut f\"ur Mathematik\\
Am Schwarzenberg-Campus 3 (E)\\
21073 Hamburg\\
Germany}
\email{christian.seifert@tuhh.de}

\begin{abstract}
  We consider the linearized Korteweg-de-Vries equations, 
  sometimes called Airy equation, on general metric graphs with edge lengths bounded away from zero.
  We show that properties of the induced dynamics can be obtained by studying boundary operators in the corresponding boundary space induced by the vertices of the graph.
  In particular, we characterize unitary dynamics and contractive dynamics. 
  We demonstrate our results on various special graphs, including those recently treated in the literature.  
\end{abstract}



\maketitle

\section{Introduction}

The Koreteweg-de-Vries equation \cite{KordeV95, Lannes13}
\[
\frac{\partial u}{\partial t}=\frac{3}{2}\sqrt{\frac{g}{\ell}}\left(\frac{\sigma}{3}\frac{\partial^3 u}{\partial x^3} +\frac{2\alpha}{3}\frac{\partial u}{\partial x}+\frac12 \frac{\partial u^2}{\partial x}\right)
\]
models shallow water waves in channels, where $u$ describes the elevation of the water w.r.t.\ the average water depth and $g$, $\ell$ , $\sigma$ and $\alpha$ are constants. 
Due to the last term on the right-hand side in becomes non-linear.
Assuming only small elevations $u$, i.e.\ $u$ is close to zero and/or long waves, i.e.\ $\frac{\partial u}{\partial x}$ is close to zero, the linear approximation neglecting the non-linearity (which is the linearization around the stationary solution $u=0$)
yields an equation of the form
\begin{equation}\label{eq:kdv_lin}
\frac{\partial u}{\partial t}= \alpha \frac{\partial^3 u}{\partial x^3} + \beta  \frac{\partial u}{\partial x}
\end{equation}
with appropriate constants $\alpha$ and $\beta$. In this paper we are going to study the evolution equation on general networks, i.e.\ metric graphs, from a functional analytic point of view.

Evolution equations (or, more general, differential operators) on metric graphs have have been intensively studied during the last two decades.
However, the focus was put on Schr\"odinger type operators and the corresponding heat and Schr\"odinger evolution equations, see \cite{BerkolaikoKuchment2013} and references therein.
Recently, also KdV-type equations on star graphs have gained interest, see \cite{DeconinckSheilsSmith16, SobAkKaJa15b, SobAkUe15a, SobAkUe15b, MugnoloNojaSeifert2017}.
Such star graphs give rise to model either singular interactions at one point, i.e.\ interface conditions, but can also be interpreted as models for junctions of channels.
The drawback of star graphs is that one only has exactly one vertex (i.e. junction), and the channels are modelled by halflines.

In this paper, we consider the linearized KdV-equation \eqref{eq:kdv_lin} on general metric graphs, i.e.\ we model a whole network, including channels of finite length.
We will describe the evolution equation in a functional analytic setup, namely in the framework of strongly continuous semigroups.
Thus, we are left to study the spacial operator describing the right-hand sinde of \eqref{eq:kdv_lin}.
The aim is to obtain a ``reasonable'' dynamics in an $L_2$-setting, meaning either unitary $C_0$-semigroups or contractive $C_0$-semigroups (which resembles the fact that the spacial derivatives appear only in odd order).
In order to do that we will employ the framework of boundary systems developped in \cite{SchSeiVoi15}.

In Section \ref{sec:NotationModel} we introduce the metric graph and the operator setup for the spacial derivatives.
Section \ref{sec:Krein} summarizes notions from Krein space theory which we will need to describe the right boundary conditions at the vertices.
We then focus on unitary and contractive dynamics in Section \ref{sec:Dynamics}.
In Section \ref{sec:LocalBC} we specialize our framework (which does not take the concrete graph structure into account) to the graph structure setting.
We end this paper by briefly listing some examples in Section \ref{sec:Examples}, where we refer to the corresponding literature, but also explaining new examples which have not been dealt with before.

\section{Notation and Model}
\label{sec:NotationModel}

Let $\Gamma = (V,E,a,b,\gamma_\l,\gamma_\r)$ be a metric graph, i.e.\ $V$ is the set of vertices of $\Gamma$, and $E$ is the set of edges of $\Gamma$.
Moreover, $a,b\from E\to [-\infty,\infty]$ such that $a_e<b_e$ for all $e\in E$ and each edge $e\in E$ is identified with the interval $(a_e,b_e)\subseteq \R$.
Let $E_\l:=\set{e\in E;\; a_e>-\infty}$ and $E_\r:=\set{e\in E;\; b_e<\infty}$ the sets of edges with finite starting and termination point, respectively, 
and let $\gamma_\l\from E_\l\to V$, $\gamma_\r\from E_\r\to V$ assign to each $e\in E_\l$ or $e\in E_\r$ the starting vertex $\gamma_\l(e)$ and the termination vertex $\gamma_\r(e)$, respectively.
Note that we do not assume $V$ or $E$ to be finite or countable.

We assume to have a positive lower bound on the edge lengths, i.e.
\begin{equation}
\label{eq:positive_edge_length}
\ell:= \inf_{e\in E} (b_e-a_e) > 0.
\end{equation}

For $k\in\set{0,1,2}$ we define $\tr_\l^k\from \bigoplus_{e\in E} W_2^{k+1}(a_e,b_e)\to \ell_2(E_\l)$ and $\tr_\r^k\from \bigoplus_{e\in E} W_2^{k+1}(a_e,b_e)\to \ell_2(E_\r)$ by
\[(\tr_\l^k u)(e) := u^{(k)}_e(a_e\rlim) \;\;(e\in E_\l), \quad (\tr_\r^k u)(e) := u^{(k)}_e(b_e\llim) \;\;(e\in E_\r).\]
Furthermore, define the trace maps $\Tr_\l\from \bigoplus_{e\in E} W_2^{3}(a_e,b_e)\to \ell_2(E_\l)^3$ and $\Tr_\r\from \bigoplus_{e\in E} W_2^{3}(a_e,b_e)\to \ell_2(E_\r)^3$ by
\[\Tr_\l u:= (\tr_\l^0 u, \tr_\l^1 u, \tr_\l^2 u),\quad \Tr_\r u:= (\tr_\r^0 u, \tr_\r^1 u, \tr_\r^2 u).\]

Let $\mathcal{H}_\Gamma:= \bigoplus_{e\in E} L_2(a_e,b_e)$ be the Hilbert space we are going to consider.

Let $(\alpha_e)_{e\in E}$ in $(0,\infty)$ be bounded and bounded away from zero, and $(b_e)_{e\in E}$ be bounded, 
and abbreviate $\alpha_\l:=(\alpha_e)_{e\in E_\l}$, $\alpha_\r:=(\alpha_e)_{e\in E_\r}$, $\beta_\l:=(\beta_e)_{e\in E_\l}$, and $\beta_\r:=(\beta_e)_{e\in E_\r}$. By the same symbol we will denote the corresponding multiplication operators in $\ell_2(E_\l)$ and $\ell_2(E_\r)$, respectively.

\begin{remark}
  One could choose $\alpha_e\in \R\setminus\{0\}$. However, for edges $e\in E$ with $\alpha_e<0$ we can just change the orientation of the edge (by setting $\tilde{\alpha}_e:=-\alpha_e$, $\tilde{\beta}_e:=-\beta_e$ and $\tilde{u}:=u(-\cdot)$).
  Thus, w.l.o.g.\ we may assume (as we will) $\alpha_e>0$ for all $e\in E$.
  
  Moreover, by scaling the variables appropriately, it would suffice to deal with the case $\beta_e = 0$ for all $e\in E$. In order to do this uniformly for all edges, one only needs boundedness of $(\beta_e)$ and $(\frac{1}{\alpha_e})$.
  However, we will keep the $\beta_e$'s possibly non-zero.
\end{remark}

\begin{definition}
  We define the \emph{minimal operator} $A_0$ in $\mathcal{H}_\Gamma$ by
  \begin{align*}
    D(A_0) & := \bigoplus_{e\in E} C_c^\infty(a_e,b_e) \\
    (A_0 u)_e & := \alpha_e \partial^3 u_e + \beta_e \partial u_e \quad(e\in E, u\in D(A_0)).
  \end{align*}
\end{definition}

\begin{definition}
  We define the \emph{maximal operator} $\hat{A}$ in $\mathcal{H}_\Gamma$ by
  \begin{align*}
    D(\hat{A}) & := \bigoplus_{e\in E} W_2^3(a_e,b_e) \\
    (\hat{A}u)_e & := \alpha_e \partial^3 u_e + \beta_e \partial u_e \quad(e\in E, u\in D(\hat{A})).
  \end{align*}
\end{definition}

Applying integration by parts we obtain the following lemma.

\begin{lemma}
  We have $-A_0^* = \hat{A}$.
\end{lemma}

Define $F\from G(\hat{A})\to \ell_2(E_\r)^3 \oplus \ell_2(E_\l)^3$ by
\[F(u,\hat{A}u):=(\Tr_\r u, \Tr_\l u).\]

\begin{lemma}
  $F$ is linear and surjective.
\end{lemma}

\begin{proof}
  Linearity of $F$ is clear.
  
  In order to show that $F$ is surjective first note that there exists $\varphi\in C_c^\infty[0,\ell)$, where $\ell$ is as in \eqref{eq:positive_edge_length}, such that
  $\varphi$ equals $1$ in a neighbourhood of $0$ and $\varphi(x) = 0$ for $x\geq \frac{\ell}{2}$.
  For $(t^0,t^1,t^2)\in \C^3$ let $u(x):= (t^0+t^1 x + \frac{1}{2}t^2 x^2)\varphi(x)$ for $x\in (0,\ell)$.
  Then there exists $c\geq 0$ (independent of $(t^0,t^1,t^2)$) such that $\norm{u}_{L_2(0,\ell)}^2+\norm{u'}_{L_2(0,\ell)}^2+\norm{u'''}_{L_2(0,\ell)}^2 \leq c\norm{(t^0,t^1,t^2)}_2^2$.
  Since $(\alpha_e)$ and $(\beta_e)$ are bounded, as a consequence, $F$ is surjective. 
\end{proof}

For $u,v\in D(\hat{A})$ we obtain by integration by parts
\begin{align}
\label{eq:integration_by_parts}
\sp{u}{\hat{A}v} + \sp{\hat{A}u}{v} = & -\sp{\begin{pmatrix} 
					      - \beta_\r & 0 & -\alpha_\r \\
                                              0 & \alpha_\r & 0 \\
                                              -\alpha_\r & 0 & 0 
                                            \end{pmatrix} \Tr_\r u}{\Tr_\r v} \nonumber\\
                                            & + \sp{\begin{pmatrix} 
					      - \beta_\l & 0 & -\alpha_\l \\
                                              0 & \alpha_\l & 0 \\
                                              -\alpha_\l & 0 & 0 
                                            \end{pmatrix} \Tr_\l u}{\Tr_\l v}.
\end{align}
Abbreviating 
\[B_{\l}:= \begin{pmatrix} 
	    - \beta_\l & 0 & -\alpha_\l \\
	    0 & \alpha_\l & 0 \\
	    -\alpha_\l & 0 & 0 
	  \end{pmatrix}, \qquad B_{\r}:= \begin{pmatrix} 
	    - \beta_\r & 0 & -\alpha_\r \\
	    0 & \alpha_\r & 0 \\
	    -\alpha_\r & 0 & 0 
	  \end{pmatrix},\]
and $\Omega\from G(\hat{A})\oplus G(\hat{A})\to \C$ by
\[\Omega\bigl((u,\hat{A} u),(v,\hat{A} v)\bigr) := \sp{\begin{pmatrix}
						  0 & 1 \\
						  1 & 0 
                                                 \end{pmatrix}\begin{pmatrix} u\\\hat{A}u\end{pmatrix}}{\begin{pmatrix} v\\\hat{A}v\end{pmatrix}},\]
and $\omega\from \ell_2(E_\r)^3\oplus \ell_2(E_\l)^3 \to \C$ by
\[\omega\bigl((x,y),(u,v)\bigr) := \sp{\begin{pmatrix} -B_\r & 0 \\ 0 & B_\l \end{pmatrix}\begin{pmatrix}x\\y\end{pmatrix} }{\begin{pmatrix} u\\v\end{pmatrix}},\]
rewriting \eqref{eq:integration_by_parts} we obtain
\begin{equation}
\label{eq:self-orthogonal} 
\Omega\bigl((u,\hat{A}u),(v,\hat{A}v)\bigr) = \omega\bigl(F(u,\hat{A}u), F(v,\hat{A}v)\bigr)
\end{equation}
for all $(u,\hat{A}u),(v,\hat{A}v)\in G(\hat{A})$.

Let $L$ be a densely defined linear operator from $\ell_2(E_\r)^3$ to $\ell_2(E_\l)^3$.
Then $A_0\subseteq A_L\subseteq \hat{A} = -A_0^*$ is defined by
\[D(A_L):= \set{u\in D(\hat{A});\; \Tr_\r u \in D(L),\, L\Tr_\r u = \Tr_\l u},\]
i.e.\ $G(A_L) = F^{-1}(G(L))$.

\section{Operators in Krein spaces}
\label{sec:Krein}

\begin{remark}
  Let $\mathcal{K}$ be a vector space and $\langle\cdot\mid \cdot\rangle$ an (indefinite) inner product on $\mathcal{K}$ such that $(\mathcal{K}, \langle\cdot\mid \cdot\rangle)$ is a Krein space.
  Then there exists an inner product $(\cdot\mid\cdot)$ on $\mathcal{K}$ such that $(\mathcal{K},(\cdot\mid\cdot))$ is a Hilbert space.
  Notions such as closedness or continuity are then defined by the underlying Hilbert space structure.
\end{remark}

\begin{definition}\label{defi:krein}
Let $\mathcal{K}_-$, $\mathcal{K}_+$ be Krein spaces, $\omega\colon \mathcal{K}_-\oplus \mathcal{K}_+ \times \mathcal{K}_-\oplus \mathcal{K}_+\to \C$ sesquilinear.
\begin{enumerate}
\item Let $X \subseteq \mathcal{K}_-\oplus \mathcal{K}_+$ be a subspace. Then $X$ is called
\emph{$\omega$-self-orthogonal} if $X = X^{\bot_\omega}$, where
\[X^{\bot_\omega}:= \{(x,y)\in \mathcal{K}_-\oplus \mathcal{K}_+;\;
\omega((x,y),(u,v)) = 0\hbox{ for all }(u,v)\in X\}.\]
\item Let $L$ be a densely defined linear operator from $\mathcal{K}_-$ to
$\mathcal{K}_+$. Then its \emph{$(\mathcal{K}_-,\mathcal{K}_+)$-adjoint}
$L^{\sharp}$ is defined by
\[
\begin{split}
D(L^{\sharp}) &:= \{y\in \mathcal{K}_+;\; \exists z\in \mathcal{K}_-\, \:
\langle Lx \mid y \rangle_+ = \langle x\mid z\rangle_-\hbox{ for all } x\in
D(L)\}\\
L^{\sharp}y&:=z.
\end{split}
\]
Clearly, $L^\sharp$ is then a linear operator from $\mathcal{K}_+$ to $\mathcal{K}_-$.
\item   Let $L$ be a linear operator from $\mathcal{K}_-$ to $\mathcal{K}_+$. Then $L$ is called a
\emph{$(\mathcal{K}_-,\mathcal{K}_+)$-contraction} if
  \[\langle Lx \mid Lx \rangle_+ \leq \langle x\mid x\rangle_- \quad\hbox{for all
}x\in D(L).\]

\item    Let $L$ be a linear operator from $\mathcal{K}_-$ to $\mathcal{K}_+$. Then $L$ is called
\emph{$(\mathcal{K}_-,\mathcal{K}_+)$-unitary} if $D(L)$ is dense, $R(L)$ is
dense, $L$ is injective, and finally $L^{\sharp} = L^{-1}$.  
\end{enumerate}
\end{definition}
If $\mathcal K_-,\mathcal K_+$ are  Hilbert spaces, then obviously $(\mathcal{K}_-,\mathcal{K}_+)$-adjoint (-contraction, -unitary) operators are the usual objects of Hilbert space operator theory.

\begin{remark}
  Let $\mathcal{K}_\pm$ be Krein spaces.
  Let $L$ be $(\mathcal{K}_-,\mathcal{K}_+)$-unitary. Then
  \begin{equation}\label{eq:defunit}
  \langle Lx \mid Ly \rangle_+ = \langle x\mid y\rangle_- \quad\hbox{for all
  }x,y\in D(L).
  \end{equation}
  However, $L$ may not be bounded.
\end{remark}

Note that $\ell_2(E_\l)^3$ equipped with $\dupa{\cdot}{\cdot}_\l\from \ell_2(E_\l)^3\times \ell_2(E_\l)^3\to\C$,
\[\dupa{(x^0,x^1,x^2)}{(y^0,y^1,y^2)}_\l:=\sp{B_\l(x^0,x^1,x^2)}{(y^0,y^1,y^2)}_{\ell_2(E_\l)^3}\]
yields a Krein space $\mathcal{K}_\l:=(\ell_2(E_\l)^3, \dupa{\cdot}{\cdot}_\l)$.
Analogously, $\ell_2(E_\r)^3$ equipped with $\dupa{\cdot}{\cdot}_\r\from \ell_2(E_\r)^3\times \ell_2(E_\r)^3\to\C$,
\[\dupa{(x^0,x^1,x^2)}{(y^0,y^1,y^2)}_\r:=\sp{B_\r(x^0,x^1,x^2)}{(y^0,y^1,y^2)}_{\ell_2(E_\r)^3}\]
yields a Krein space $\mathcal{K}_\r:=(\ell_2(E_\r)^3, \dupa{\cdot}{\cdot}_\r)$.

\section{Dynamics}
\label{sec:Dynamics}

We now study different types of dynamics for the equation.

\subsection{Generating unitary groups}

We are now interested in generators of unitary groups.
By Stone's Theorem this is equivalent to look for skew-self-adjoint realizations of $A_L$.

\begin{theorem}\label{thm:skew-self-adjoint}
  Let $L$ be a linear operator from $\ell_2(E_\r)^3$ to $\ell_2(E_\l)^3$ such that $D(L)$ and $R(L)$ are dense.
  Then $A_L$ is skew-self-adjoint if and only if $L$ is $(\mathcal{K}_\r,\mathcal{K}_\l)$-unitary.
\end{theorem}

\begin{proof}
  By \cite[Corollary 2.3 and Example 2.7.(b)]{SchSeiVoi15}, see also \cite[Theorem 3.7]{MugnoloNojaSeifert2017},
  we have to show that $G(A_L)$ is $\Omega$-self-orthogonal if and only if $G(L)$ is $\omega$-self-orthogonal. But this is an easy consequence of \eqref{eq:self-orthogonal} and the definition of $A_L$.
\end{proof}

\subsection{Generating contraction semigroups}

Instead of unitary dynamics we now ask for generators of contraction semigroups.

\begin{theorem}
\label{thm:contractive}
  Let $L$ be a densely defined closed linear operator from $\ell_2(E_\r)^3$ to $\ell_2(E_\l)^3$.
  Then $A_L$ is the generator of a semigroup of contractions if and only if $L$ is $\bigl(\mathcal{K}_\r,\mathcal{K}_\l\bigr)$-contractive
  and $L^\sharp$ is $\bigl(\mathcal{K}_\l,\mathcal{K}_\r\bigr)$-contractive. 
\end{theorem}

\begin{proof}
  Note that since $L$ is closed, also $A_L$ is closed.
  
  First, it is easy to see that $A_L^*$ is given by
  \begin{align*}
    D(A_L^*) & = \Bigl\{u\in D(A_0^*);\; \Tr_\l u \in D(L^\sharp),\, L^\sharp \Tr_\l u = \Tr_\r u\Bigr\},\\
    A_L^* u & = A_0^* u.
  \end{align*}

  For $u\in D(A_L)$ we compute
  \begin{align*}
    2\Re \sp{A_L u}{u} & = \Omega\bigl((u,A_L u),(u,A_L u)\bigr) \\
    & = \omega\bigl(F(u,A_Lu),F(u,A_Lu)\bigr) \\
    & = \omega\bigl((\Tr_\r u, L\Tr_\r u),(\Tr_\r u, L\Tr_\r u)\bigr) \\
    & = -\dupa{\Tr_\r u}{\Tr_\r u}_\r + \dupa{L \Tr_\r u}{L \Tr_\r u}_\l.
  \end{align*}
  Hence, $A_L$ is dissipative, i.e.\ $\Re \sp{A_L u}{u}\leq 0$ for all $u\in D(A_L)$, if and only if $L$ is $\bigl(\mathcal{K}_\r,\mathcal{K}_\l\bigr)$-contractive.
  
  Similarly, $A_L^*$ is dissipative if and only if $L^\sharp$ is $\bigl(\mathcal{K}_\l,\mathcal{K}_\r\bigr)$-contrac\-tive.
  
  Thus, the Lumer-Phillips Theorem in Hilbert spaces yields the assertion.
\end{proof}

\section{Local boundary conditions}
\label{sec:LocalBC}

So far, we did not take into account the graph structure.
Now, we ask for boundary conditions at each vertex $v\in V$ separately. 
For $v\in V$ let
\[E_{\l,v}:= \set{e\in E_\l;\; \gamma_\l(e) = v},\quad E_{\r,v}:= \set{e\in E_\r;\; \gamma_\r(e) = v}.\]
Then $\ell_2(E_{\l,v})$ equipped with $\dupa{\cdot}{\cdot}_{\l,v}:=\dupa{\cdot}{\cdot}_{\l}|_{\ell_2(E_{\l,v})^3\times \ell_2(E_{\l,v})^3}$
 and $\ell_2(E_{\r,v})$ equipped with $\dupa{\cdot}{\cdot}_{\r,v}:=\dupa{\cdot}{\cdot}_{\r}|_{\ell_2(E_{\r,v})^3\times \ell_2(E_{\r,v})^3}$ 
yield Krein spaces $\mathcal{K}_{\l,v}$ and $\mathcal{K}_{\r,v}$, respectively, such that
\[\dupa{\cdot}{\cdot}_{\l} = \sum_{v\in V} \dupa{\cdot}{\cdot}_{\l,v}\]
and analogously for $\dupa{\cdot}{\cdot}_{\r}$.
	  
For $v\in V$ let $L_v$ be a densely defined linear operator from $\ell_2(E_{\r,v})^3$ to $\ell_2(E_{\l,v})^3$, and define
$A_0\subseteq A_{(L_v)_{v\in V}}\subseteq \hat{A} = -A_0^*$ by
\begin{align*}
  D(A_L) := \Bigl\{u\in D(\hat{A});\; \forall\,v\in V: &\; (\Tr_\r u)|_{E_{\r,v}} \in D(L_v),\\
  &\;  L_v \bigl((\Tr_\r u)|_{E_{\r,v}}\bigr) = \Tr_\l u|_{E_{\l,v}}\Bigr\}.
\end{align*}

For the case of unitary dynamics we obtain the following corollary.

\begin{corollary}
  For $v\in V$ let $L_v$ be a densely defined linear operator from $\ell_2(E_{\r,v})^3$ to $\ell_2(E_{\l,v})^3$.
  Then $A_{(L_v)_{v\in V}}$ is the generator of a unitary group if and only if $L_v$ is 
  $(\mathcal{K}_{\r,v},\mathcal{K}_{\l,v})$-unitary for all $v\in V$.
\end{corollary}

\begin{proof}
  We show that $(L_v)_{v\in V}$ is $(\mathcal{K}_\r,\mathcal{K}_\l)$-unitary if and only if $L_v$ is $(\mathcal{K}_{\r,v},\mathcal{K}_{\l,v})$-unitary for all $v\in V$.
  Then the result follows from Theorem \ref{thm:skew-self-adjoint}.
  Note that $(L_v)_{v\in V}$ acts as a block-diagonal operator with block $L_v$ from $\ell_2(E_{\r,v})^3$ to $\ell_2(E_{\l,v})^3$ for all $v\in V$.
  Hence, clearly, $(L_v)_{v\in V}$ is densely defined with dense range if and only if $L_v$ is densely defined with dense range for all $v\in V$.
  Moreover, $(L_v)$ is injective if and only if $L_v$ is injective for all $v\in V$.
  Since $(L_v)_{v\in V}^\sharp = (L_v^\sharp)_{v\in V}$ we also obtain $(L_v)^\sharp = (L_v)^{-1}$ if and only if $L_v^\sharp = L_v^{-1}$ for all $v\in V$.
\end{proof}

Analogously, in case of contractive dynamics we have the following result.

\begin{corollary}
  For $v\in V$ let $L_v$ be a densely defined linear operator from $\ell_2(E_{\r,v})^3$ to $\ell_2(E_{\l,v})^3$ such that $L:=(L_v)_{v\in V}$ is closed.
  Then $A_{(L_v)_{v\in V}}$ is the generator of a semigroup of contractions if and only if $L_v$ is 
  $(\mathcal{K}_{\r,v},\mathcal{K}_{\l,v})$-contractive and $L_v^\sharp$ is $(\mathcal{K}_{\l,v},\mathcal{K}_{\r,v})$-contractive
  for all $v\in V$.
\end{corollary}

\begin{proof}
  We show that $(L_v)_{v\in V}$ is $(\mathcal{K}_\r,\mathcal{K}_\l)$-con\-tractive and $((L_v)_{v\in V})^\sharp$ 
  is $(\mathcal{K}_\l,\mathcal{K}_\r)$-con\-tractive if and only if $L_v$ is 
  $(\mathcal{K}_{\r,v},\mathcal{K}_{\l,v})$-con\-tractive and
  $L_v^\sharp$ is $(\mathcal{K}_{\l,v},\mathcal{K}_{\r,v})$-con\-tractive
  for all $v\in V$.
  Then the result follows from Theorem \ref{thm:contractive}.  
  Again, $(L_v)_{v\in V}$ acts as a block-diagonal operator with block $L_v$ from $\ell_2(E_{\r,v})^3$ to $\ell_2(E_{\l,v})^3$ for all $v\in V$.
  Hence, $(L_v)$ is $(\mathcal{K}_\r,\mathcal{K}_\l)$-con\-tractive if and only if $L_v$ is $(\mathcal{K}_{\r,v},\mathcal{K}_{\l,v})$-con\-tractive for all $v\in V$.
  Since $(L_v)^{\sharp} = (L_v^\sharp)$, we have that $(L_v)^\sharp$ is $(\mathcal{K}_\l,\mathcal{K}_\r)$-con\-tractive if and only if 
  $L_v^\sharp$ is $(\mathcal{K}_{\l,v},\mathcal{K}_{\r,v})$-con\-tractive
  for all $v\in V$.  
\end{proof}

\section{Examples}
\label{sec:Examples}

In this section we will specialize to particular examples of graphs. For those secial cases already treated in the literture we just explain the setup.
We ask the reader to go to the corresponding references for more details in these cases.

\subsection{Two semi-infinite edges}

The case of two semi-infinite edges correspond to $\abs{V} = 1$, $\abs{E} = 2$ and the two edges correspond to the intervals $(-\infty,0]$ and $[0,\infty)$.

\begin{center}
\begin{tikzpicture}[scale=0.75]
  \foreach \a in {0,180} {
    \draw[rotate=\a] (0,0)--(3,0);
    \draw[rotate=\a, dashed] (3,0)--(5,0);
  }  
  \draw[fill] (0,0) circle(0.05) node[below]{$v$};
\end{tikzpicture}
\end{center}

If the coefficients $(\alpha_e)$ and $(\beta_e)$ are constant, we can interpret the equation as the linearized KdV equation on the real line with a generalized point interaction at $0$ (which corresponds to the vertex $v$).
This situation was considered in \cite{DeconinckSheilsSmith16}.

As a particular example, if $\alpha_e = 1$ and $\beta_e = 0$ for all $e$, then $L:=\begin{pmatrix}
										    1 & 0 & 0 \\
										    \sqrt{2} & 1 & 0 \\
										    1 & \sqrt{2} & 1
                                                                                   \end{pmatrix}$ yields a unitary dynamics,
since $L\from \mathcal{K}_\r\to \mathcal{K}_\l$ is $(\mathcal{K}_\r,\mathcal{K}_\l)$-unitary, i.e. $L$ is bijective and
\[\dupa{Lx}{Ly}_{\l} = \dupa{x}{y}_\r \quad(x,y\in \mathcal{K}_\r).\]

\subsection{Star graphs}

The special case of star graphs were considered in \cite{SobAkUe15a, SobAkUe15b, SobAkKaJa15b, MugnoloNojaSeifert2017}.
Here, $\abs{V} = 1$, each edge corresponds to a semi-infinite interval, $(-\infty,0]$ or $[0,\infty)$, say, and each edge is adjacent to $v$ with its endpoint corresponding to the value $0$ for the interval.

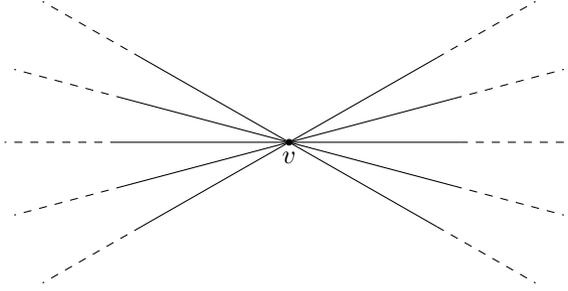
\begin{center}
\begin{tikzpicture}[scale=0.75]
  \foreach \a in {-30,-15,0,15,30,150,165,180,195,210} {
    \draw[rotate=\a] (0,0)--(3,0);
    \draw[rotate=\a, dashed] (3,0)--(5,0);
  }  
  \draw[fill] (0,0) circle(0.05) node[below]{$v$};
\end{tikzpicture}
\end{center}

Star graphs generalize graphs with two semi-infinite edges to more than two edges.

\subsection{A Loop}

Let $\Gamma$ be a loop, i.e.\ $\abs{V} = \abs{E} = 1$, and the edge corresponds to the interval $[0,1]$ and both endpoints of the edge are attached to the vertex $v$.

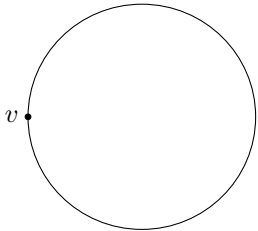
\begin{center}
\begin{tikzpicture}[scale=0.75]
  \draw (0,0) circle(2); 
  \draw[fill] (-2,0) circle(0.05) node[left]{$v$};
\end{tikzpicture}
\end{center}

Here, we can model generalized periodic boundary conditions.
Indeed, consider $L$ to be represented by the identity matrix in the usual basis. Then this results in periodic boundary conditions
\[u^{(k)}(0\rlim) = u^{(k)}(1\llim) \quad(k\in\{0,1,2\}).\]
Moreover, $L$ becomes Krein space unitary in this case, so the dynamics is unitary.


\begin{thebibliography}{99}

\bibitem{BerkolaikoKuchment2013}
G. Berkolaiko, P. Kuchment:
Introduction to Quantum Graphs.
Mathematical Surveys and Monographs {\bf 186}, AMS (2013).

\bibitem{DeconinckSheilsSmith16}  
B. Deconinck, N. E. Sheils, D. A. Smith:
\textit{The linear KdV equation with an interface}. 
Comm. Math. Phys. {\bf 347}(2), 489--509 (2016).


\bibitem{Lannes13}
D. Lannes:
The Water Waves Problem: Mathematical Analysis and Asymptotics.
Mathematical Surveys and Monographs {\bf 188}, AMS (2013).

\bibitem{KordeV95}
D. J. Korteweg and G. de Vries.
\textit{On the change of form of long waves advancing in a rectangular canal, and on a new type of long stationary waves}.
The London, Edinburgh, and Dublin Phil. Mag. and Journal of Sci. {\bf 39}, 422--443 (1895).

\bibitem{MugnoloNojaSeifert2017}
D.~Mugnolo, D.~Noja, C.~Seifert:
\textit{Airy-type evolutuion equations of star graphs}.
submitted. arXiv-Preprint 1608.01461.

\bibitem{SchSeiVoi15}
C.~Schubert, C.~Seifert, J.~Voigt, and M.~Waurick:
\textit{Boundary systems and (skew-) self-adjoint operators on infinite metric graphs}.
Math.\ Nachr. {\bf 288}, 1776--1785 (2015).

\bibitem{SobAkUe15a} 
Z. A. Sobirov, M. I. Akhmedov, H. Uecker:
\textit{Cauchy problem for the linearized KdV equation on general metric star graph}.
Nanosystems {\bf 6}, 198--204 (2015).

\bibitem{SobAkUe15b} 
Z. A. Sobirov, M. I. Akhmedov, H. Uecker:
\textit{Exact solution of the Cauchy problem for the linearized KdV equation on metric star graph}.
Uzbek Math. J. {\bf 3}, 143--154 (2015).

\bibitem{SobAkKaJa15b} 
Z. A. Sobirov, M. I. Akhmedov, O. V. Karpova, B. Jabbarova:
\textit{Linearized KdV equation on a metric graph}.
Nanosystems {\bf 6}, 757--761, (2015).  


\end{thebibliography}
\end{document}